\documentclass[12pt,3p,sort&compress]{elsarticle}

\usepackage{amsmath,amsthm,amssymb}
\usepackage[mathlines,pagewise]{lineno}
\usepackage{listings}
\usepackage[draft]{varioref}
\usepackage{caption}
\usepackage{url}

\DeclareMathOperator{\Fact}{Fact}
\DeclareMathOperator{\odeg}{odeg}
\renewcommand{\epsilon}{\varepsilon}

\renewcommand{\epsilon}{\varepsilon}

\newtheorem{theorem}{Theorem}

\newtheorem{lemma}[theorem]{Lemma}
\newtheorem{corollary}[theorem]{Corollary}

\theoremstyle{definition}

\newtheorem{example}[theorem]{Example}

\begin{document}

\sloppy

\begin{frontmatter}

\title{On the Lie complexity of Sturmian words}
 
\author{Alessandro De Luca}
 \ead{alessandro.deluca@unina.it}
 \address{DIETI, Universit\`a degli Studi di Napoli Federico II, Italy}
 
\author{Gabriele Fici\corref{cor1}\fnref{grant}}
\ead{gabriele.fici@unipa.it}
\address{Dipartimento di Matematica e Informatica, Universit\`a di Palermo, Italy}

\cortext[cor1]{Corresponding author.}
\fntext[grant]{Supported by MIUR PRIN 2017 Project 2017K7XPAN.}

\journal{Theoretical Computer Science}

\begin{abstract}
 Bell and Shallit recently introduced the Lie complexity of an infinite word $s$ as the function counting for each length the number of conjugacy classes of words whose elements are all factors of $s$. They proved, using algebraic techniques, that the Lie complexity is bounded above by the first difference of the factor complexity plus one; hence, it is uniformly bounded for words with linear factor complexity, and, in particular, it is at most $2$ for Sturmian words, which are precisely the words with factor complexity $n+1$ for every $n$. In this note, we provide an elementary combinatorial proof of the result of Bell and Shallit and give an exact formula for the Lie complexity of any Sturmian word.
\end{abstract}


\end{frontmatter}

\section{Introduction}

The factor complexity $p_w$ of an infinite word $w$ is the integer function that counts, for every nonnegative integer $n$, the number of distinct factors of length $n$ occurring in $w$. This notion is widely used in the combinatorial investigation of infinite sequences. For example, it is used in the definition of topological entropy of a symbolic dynamical system. 

A fundamental result of Morse and Hedlund~\cite{MoHe38} is that any aperiodic right-infinite word has factor complexity at least $n+1$ for every $n$. Sturmian words are aperiodic words with minimal factor complexity,  i.e., they have factor complexity equal to $n+1$ for every $n$ (in particular they have two factors of length $1$, i.e., they are binary words). 

In the literature, other complexity functions have been introduced. To cite a few, abelian complexity~\cite{CRSZ2010}, 
$k$-abelian complexity~\cite{DBLP:journals/jct/KarhumakiSZ13}, 
arithmetic complexity~\cite{DBLP:conf/dlt/AvgustinovichFF00,DBLP:journals/tcs/CassaigneF07},
maximal pattern complexity~\cite{kamae_zamboni_2002}, 
cyclic complexity~\cite{CaFiScZa15}, 
binomial complexity~\cite{DBLP:journals/tcs/RigoS15}, 
window complexity~\cite{window}, periodicity complexity~\cite{DBLP:journals/ijac/MignosiR13}, etc. 

Recently, Bell and Shallit~\cite{BS22} introduced the notion of \emph{Lie complexity} of an infinite word $w$ as the integer function whose value at $n$ is the number of conjugacy classes (under cyclic shift) of factors of length $n$ of $w$ with the property that every element in the conjugacy class occurs as a factor in $w$. We call such a conjugacy class a \emph{Lie class} of factors of $w$.

Bell and Shallit proved  the following result:

\begin{theorem}[\cite{BS22}]\label{thm:BS} 
Let $\Sigma$ be a finite alphabet, let $w$ be a right-infinite word over  $\Sigma$, and let $L_w:\mathbb{N}\mapsto \mathbb{N}$ be the Lie complexity function of $w$. Then for each $n\geq 1$ we have
\[L_w(n)\leq p_w(n)-p_w(n-1)+1.\]
\end{theorem}

Hence, the Lie complexity is uniformly bounded for words  with linear factor complexity, and, in particular, it is bounded by $2$ for Sturmian words.

The proof of the previous theorem given in~\cite{BS22} is purely algebraic. In this note, we provide an elementary combinatorial proof of this result.

We then  give an exact formula for the Lie complexity of any Sturmian word of slope $\alpha$ in terms of the continued fraction expansion of $\alpha$.
For a general introduction to Sturmian words the reader is pointed to~\cite{LothaireAlg}.

\section{A combinatorial proof for the bound on the Lie complexity}

For all $n\geq 0$, let $\Fact_w(n)$ denote the set of factors of length $n$ of $w$, so that $p_w(n)= \#\Fact_w(n)$.
Recall that the \emph{Rauzy graph} of order $n\geq 1$ for $w$, denoted by $\Gamma_w(n)$, is the directed graph with set of vertices $\Fact_w( n-1)$ and set of edges $\Fact_{w}(n)$ such that an edge $e \in \Fact_{w}(n)$  starts at  vertex $v$ and ends at a vertex $v'$ if and only if $v$ is a prefix of $e$ and $v'$ is a suffix of $e$.

Recall that in a directed graph, a (simple) \emph{cycle} is a walk that starts and ends in the same vertex and no other vertex is repeated. For our purposes, we identify cycles having the same sets of vertices (and edges).

\begin{lemma}\label{lem:lc}
Lie classes of factors of length $n$ correspond exactly to cycles  whose lengths divide $n$ in $\Gamma_w(n)$.
\end{lemma}

\begin{proof}
Suppose that all cyclic shifts of $u=a_1\cdots a_n$ are factors of $w$. Then such shifts correspond to consecutive edges of a cycle in the Rauzy graph; if they are all distinct, i.e., $u$ is primitive, then clearly the cycle has length $n$. Otherwise we can write $u=v^{n/d}$ for some $v$, with $d$ the smallest index such that $u=a_{d+1}\cdots a_na_1\cdots a_d$.

Conversely, let $u_1,\ldots,u_d$ be consecutive edges of a cycle, with $d\vert n$, and set
\[\begin{split}u_1 &=a_1a_2\cdots a_{n-1}x_1,\\
u_2 &=a_2\cdots a_{n-1}x_1x_2,\\
&\ \, \vdots\\
u_d &=a_d\cdots a_{n-1}x_1\cdots x_d
\end{split}\]
for letters $a_1,\ldots, a_{n-1}$ and $x_1,\ldots, x_d$. Since the last edge $u_d$ returns to the starting vertex $a_1\cdots a_{n-1}$, the word
$a_1\cdots a_{n-1}x_1\cdots x_d$ has $a_1\cdots a_{n-1}$ as a suffix as well as a prefix. This implies that all its factors $u_1,\ldots, u_d$ have $d$ as a period, so that they are all the cyclic shifts of $u_1$.
\end{proof}

In view of the previous lemma, we say that a cycle in the Rauzy graph $\Gamma_w(n)$ is a \emph{Lie cycle} if its length divides $n$.
Thus, $L_w(n)$ is the number of Lie cycles in $\Gamma_w(n)$, whereas $p_w(n)$ and $p_w(n-1)$ are the numbers of edges and vertices, respectively.

For a vertex $v$, we let $\odeg(v)$ denote the \emph{out-degree} of $v$, i.e., the number of distinct edges leaving $v$.

\begin{proof}[Proof of Theorem~\ref{thm:BS}]
We first observe that two Lie cycles may share one or more vertices but cannot share edges, since conjugacy classes are disjoint. As a consequence, if a vertex belongs to $k$ different Lie cycles, its out-degree is at least $k$.

We show that in $\Gamma_w(n)$ there exists a set $\mathcal L$ of $L_w(n)-1$ edges such that every vertex of $\Gamma_w(n)$ has an outgoing edge not belonging to $\mathcal L$; this proves that the number of edges minus the number of vertices is at least $L_w(n)-1$, whence the claimed inequality $L_w(n)\leq p_w(n)-p_w(n-1)+1$.

Consider a walk on $\Gamma_w(n)$ visiting at least one edge for each Lie cycle ($w$ itself provides an example of such a walk). With the possible exception of the last one visited, every Lie cycle must contain a vertex with out-degree at least $2$. Suppose $v$ is such a vertex, and let $k\geq 1$ be the number of Lie cycles containing $v$, so that $\odeg(v)\geq k$. Then, since the walk visits all Lie cycles in $\Gamma_w(n)$, at least one of the following cases occurs:
\begin{enumerate}
    \item one of the $k$ cycles is the last one visited by the walk;
    \item $\odeg(v)\geq k+1$;
    \item at least one of the $k$ cycles contains a vertex $v'\neq v$ with $\odeg(v')\geq 2$.
\end{enumerate}
Therefore, we can define $\mathcal L$ as follows: for each of the first $L_w(n)-1$ Lie cycles, we choose an edge belonging to the same Lie cycle and leaving from a vertex with out-degree at least $2$, with the requirement that each of these vertices has at least one outgoing edge which is not chosen. This choice for $\mathcal L$ ensures that each vertex in $\Gamma_w(n)$ has at least one outgoing edge not belonging to $\mathcal L$, as required.
\end{proof}

\section{A formula for the Lie complexity of Sturmian words} 

A Sturmian word $s=s_{\alpha,\rho}$ over $\Sigma=\{0,1\}$ can be defined by taking an irrational number $0<\alpha<1$ (called \emph{slope}) and a real number $\rho$ (called \emph{intercept}) and defining for each $n\geq 0$
\[s_{\alpha,\rho}(n)=\lfloor \alpha(n+1) + \rho \rfloor - \lfloor \alpha n + \rho \rfloor\]

As is well known, any two Sturmian words $s=s_{\alpha,\rho}$ and $s'=s'_{\alpha,\rho'}$ with the same slope have the same factors. Therefore, one often considers the \emph{characteristic}  Sturmian word of slope $\alpha$, which is the word $s_{\alpha,\alpha}$. 

Let $[0;d_1+1,d_2,\ldots,d_n,\ldots]$ be the continued fraction expansion of $\alpha$. We will assume that $11$ is not a factor of $s_{\alpha,\alpha}$, which corresponds to assuming $d_1>0$, i.e., $\alpha<1/2$. The other case, i.e., when $11$ is a factor of $s_{\alpha,\alpha}$ for $\alpha=[0;1,d_2,d_3,\ldots]$, can be reduced to the previous one by considering the characteristic Sturmian word obtained by exchanging the two letters, which has slope $\alpha'=[0;d_2+1,d_3,\ldots]$. 

The characteristic Sturmian word $s=s_{\alpha,\alpha}$ is the limit of the sequence of finite words $s_{-1}=1$, $s_0=0$ and $s_{n}=s_{n-1}^{d_{n}}s_{n-2}$ 
for $n>0$. The words $s_k$, $k\geq 0$, are called \emph{standard prefixes} of $s$. 

For each $k\geq 0$, the length of $s_k$ is equal to $q_k$, the denominator of the $k$-th \emph{convergent} $p_k/q_k=[0;d_1+1,d_2,\ldots,d_k]$ (we assume $q_0=1$). We will also need, when $d_k>1$, the denominators $q_{k,\ell}$ of the $k$-th \emph{semiconvergents}  $p_{k,\ell}/q_{k,\ell}=[0;d_1+1,d_2,\ldots,d_{k-1},\ell]$, $1\leq \ell < d_k$. The words $s_{k,\ell}=s_{k-1}^{\ell}s_{k-2}$ of length $q_{k,\ell}$ are sometimes called \emph{semistandard} prefixes of $s$.

Let $S$ denote the set of standard or semistandard prefixes of $s$. For every word $v\in S$ of length at least $2$, one has $v=uab$,
where $ab\in\{01,10\}$ and the word $u$, called a \emph{central prefix}, is a \emph{bispecial factor} of $s$. Recall that a factor $u$ of $s$ is left (resp.~right) \emph{special} if both $0u,1u$ (resp.~both $u0,u1$) are factors of $s$ and bispecial if it is both left special and right special. Notice that since a Sturmian word has $n+1$ factors of length $n$, it must have exactly one left (resp.~right) special factor of each length $n$, and this must therefore be a prefix (resp.~suffix) of a bispecial factor.

The following result follows from~\cite[Lemma 9]{CaFiScZa15}.

\begin{lemma}\label{cyclic}
Let $s$ be a Sturmian word and $w$ a primitive factor of $s$ of length at least $2$. Then all conjugates of $w$ are factors of $s$ if and only if $w$ is a conjugate of an element of $S$.
\end{lemma}

The best known example of a Sturmian word is the Fibonacci word $f=0100101001001\cdots$, which can be defined as the fixed point of the morphism sending $0$ to $01$ and $1$ to $0$. The Fibonacci word is intimately related to the well-known sequence of Fibonacci numbers:  $F_1 = 1$, $F_2 = 1$, and $F_n =F_{n-1}+F_{n-2}$ for $n\geq 2$. The precise relation is the following: $f$ is the characteristic Sturmian word $s_{1/\phi^2,1/\phi^2}$, where $\phi=(1+\sqrt 5)/2$ is the golden ratio. Since $1/\phi^2=[0;2,\overline{1}]$, we have that for the Fibonacci word $d_n=1$ for every $n$ and the sequence $q_n=F_{n+2}$ is the sequence of denominators of the convergents of $1/\phi^2$. The standard prefixes of $f$ (of length $F_n$) are the Fibonacci finite words $1$, $0$, $01$, $010$, $01001$, etc.

In Example 7.4 of~\cite{BS22}, the authors looked at the Lie complexity $L_f$ of the Fibonacci word $f$ and showed that
\[
L_f(n)=
\begin{cases}
1 \text{, if $n =0$ or $n=F_k$ for $k \geq 4$ or $n=F_k+F_{k-3}$ for $k \geq 4$;}\\
2 \text{, if $n =1,2$;}\\
0  \text{, otherwise.}
\end{cases}
\]
Notice that $F_k+F_{k-3}=F_{k-1}+F_{k-2}+F_{k-3}=2F_{k-1}$.

 The main result of this section is the following:

\begin{theorem}\label{thm:formula}
The Lie complexity of any Sturmian word $s$ of slope $\alpha<1/2$ is: 
\[
L_s(n)=
\begin{cases}
1 \text{, if $n =0$ or $n=q_{k,\ell}$ for $k \geq 2$ or $n=mq_{k}$ for $1\leq m\leq d_{k+1}+1$ and $k \geq 1$;}\\
2 \text{, if $n =1,2,\ldots,q_1$;}\\
0  \text{, otherwise.}
\end{cases}
\]
\end{theorem}

\begin{lemma}\label{main2}
Let $s$ be a Sturmian word and $w$ a factor of $s$ of length at least $2$. If all conjugates of $w$ are factors of $s$, then $w$ is a power of a conjugate of an element of $S$.
\end{lemma}

\begin{proof}
If $w=v^m$, $v$ primitive, and all conjugates of $w$ are factors of $s$, then in particular all conjugates of $v$ are factors of $s$, hence by Lemma\ref{cyclic}, $v$ is a conjugate of  an element of $S$.
\end{proof}

\begin{example}
The converse is not true. Consider the Fibonacci word $f=010010100100101001\cdots$. The factor $w=(010)^3$ is a power of the standard prefix $010$, yet no other conjugate of $w$ is a factor of $f$.
\end{example}

The following result is due to Damanik and Lenz~\cite[Thm.~4]{DBLP:journals/ejc/DamanikL03} (see also~\cite{peltomaki2016privileged}). Recall that the \emph{index} of a factor $v$ of $s$ is the largest integer $n$ such that $v^n$ is a factor of $s$.

\begin{theorem}[\cite{DBLP:journals/ejc/DamanikL03}]\label{thm:dl}
Let $s$ be a Sturmian word.
\begin{itemize}
\item All conjugates of the standard prefix $s_1$ have index $d_2+1$;
    \item For every $k\geq 2$, the set of indexes of all conjugates of the standard prefix $s_k$ is $\{d_{k+1}+1,d_{k+1}+2\}$;
    
    \item For every $k\geq 2$, the set of indexes of all conjugates of a semistandard prefix $s_{k,\ell}$ is $\{1,2\}$.
\end{itemize}
\end{theorem}

\begin{corollary}\label{cor:dl}
Let $s$ be a Sturmian word.
\begin{itemize}
    \item For every $k\geq 1$ and $1\leq m\leq d_{k+1}+1$, all  conjugates of $s_k^m$ are factors of $s$, but not all conjugates of $s_k^{d_{k+1}+2}$  are factors of $s$;
    
    \item For every $k\geq 2$, all  conjugates of $s_{k,\ell}$ are factors of $s$, but not all conjugates of $s_{k,\ell}^2$   are  factors of $s$.
\end{itemize}
\end{corollary}

We are now able to give the proof of Theorem~\ref{thm:formula}.

 \begin{proof}
The assertion is trivially verified for $n=0$, as well as for $1\leq n\leq q_1$ since the $n+1$ factors of $s$ of length $n$ are $0^n$ and the $n$ conjugates of $0^{n-1}1$.

Let then $n>q_1$, and suppose $L_s(n)>0$, so that there exists a factor $w$ of length $n$ such that all conjugates of $w$ are factors of $s$. By Lemma~\ref{main2}, there exists $v\in S$ and $m\geq 1$ such that all conjugates of $v^m$ are factors of $s$. Since $n>q_1$, either $v=s_{k-1}$ or $v=s_{k,\ell}$ for  some $k\geq 2$. By~Corollary~\ref{cor:dl}, $n=q_{k,\ell}$ for $k \geq 2$ or $n=mq_{k}$ for $1\leq m\leq d_{k+1}+1$ and $k \geq 1$.

To conclude the proof, we must show that $L_s(n)\leq 1$ for $n>q_1$, i.e., that the prefix $v$ is uniquely determined by $n$. For $k\geq 1$ and $1\leq\ell<d_{k+1}$, by definition one has the following:
\begin{equation}\label{eq:q}
    q_{k+1}=d_{k+1}q_{k}+q_{k-1},\quad
q_{k+1,\ell}=\ell q_{k}+q_{k-1}\,.
\end{equation}
In particular, the sequence $(q_k)$ is strictly increasing; let then $k\geq 2$ be such that $q_k\leq n<q_{k+1}$, where $L_s(n)>0$. By the above argument, the possible values for $n$ are
\begin{enumerate}
    \item $mq_k$, for $1\leq m\leq d_{k+1}$,
    \item $(d_k+1)q_{k-1}=q_k+q_{k-1}-q_{k-2}$,
    \item $q_{k+1,\ell}$, for $1\leq\ell<d_{k+1}$.
\end{enumerate}
In view of \eqref{eq:q}, these are all distinct, so that the corresponding value for $|v|$ (respectively $q_k$, $q_{k-1}$, and $q_{k+1,\ell}$) is well defined and uniquely determined.
 \end{proof}

\section*{Acknowledgments}

We thank all the participants of Jason Bell's seminar~\cite{Bell21}, in particular: Jason Bell, Jeffrey Shallit, Christophe Reutenauer and Narad Rampersad.

%

\end{document}